\documentclass[12pt,letterpaper,draftcls,onecolumn]{IEEEtran}
\usepackage{cite}
\usepackage{graphicx}
\usepackage{psfrag}
\usepackage{subfigure}
\usepackage{url}
\usepackage{stfloats}
\usepackage{amsmath}
\usepackage{array}
\usepackage{fancyhdr}
\usepackage{latexsym}
\usepackage{amssymb}
\usepackage{stmaryrd}
\usepackage{qsymbols}
\usepackage{ulsy}

\newtheorem{lemma}{Lemma}

\begin{document}
\title{Accurate Performance Analysis of Opportunistic Decode-and-Forward Relaying\thanks{This work was supported by the Qatar National Research Fund (A member of Qatar Foundation). Kamel Tourki is with the Electrical and Computer Engineering Program, Texas A\&M University at Qatar, Education City, Doha, Qatar. Hong-Chuan Yang is with the Department of Electrical and Computer Engineering, University of Victoria, BC, Canada. Mohamed-Slim Alouini is with the Electrical Engineering Program, Division of Physical Sciences and Engineering, King Abdullah University of Science and Technology (KAUST), Thuwal, Mekkah Province, Kingdom of Saudi Arabia. The corresponding author is Kamel Tourki (kamel.tourki@qatar.tamu.edu).}}
\author{Kamel~Tourki,~\IEEEmembership{Member,~IEEE,}
        Hong-Chuan~Yang,~\IEEEmembership{Senior Member,~IEEE,}
        Mohamed-Slim~Alouini,~\IEEEmembership{Fellow,~IEEE}.
        }
\maketitle
\begin{abstract}
In this paper, we investigate an opportunistic relaying scheme where the selected relay assists the source-destination (direct) communication. In our study, we consider a regenerative opportunistic relaying scheme in which the direct path can be considered unusable, and takes into account the effect of the possible erroneously detected and transmitted data at the best relay. We first derive statistics based on exact probability density function (PDF) of each hop. Then, the PDFs are used to determine accurate closed form expressions for end-to-end bit-error rate (BER) of binary phase-shift keying (BPSK) modulation. Furthermore, we evaluate the asymptotical performance analysis and the diversity order is deduced. Finally, we validate our analysis by showing that performance simulation results coincide with our analytical results over different network architectures.
\end{abstract}
\begin{center}
    \textbf{Keywords}:
\end{center}
Cooperative diversity, Opportunistic Regenerative relaying, Performance analysis.
\newpage
\section{Introduction}
\label{Intro}In many wireless applications, users may not be able to support multiple antennas due to size, complexity, power, or other constraints. The wireless medium brings along its unique challenges such as fading and multiuser interference. This can be mitigated with cooperative diversity \cite{LaW,SEA1,SEA2}, which is becoming very attractive for small-size, antenna-limited wireless devices. Opportunistic relaying (OR) technique has been proposed where only the best relay from a set of $K$ available candidate relays is selected to cooperate~\cite{Bal06,TAH09,Kal09,Bal10}. With this technique, the selection strategy is to choose the relay with the best equivalent end-to-end channel gain which is calculated as the minimum of the channel gains of the first and the second hops under decode-and-forward (DF) protocol or with the best harmonic mean of both channel gains under amplify-and-forward (AF) protocol. However, some works have chosen the best relay-destination link as possible selection criteria~\cite{TAH09,IA09}.\\
Previous works have largely focused on information theoretic aspects of OR and derived outage performance results of such systems. Some of these analysis are accurate only at high signal to noise ratio (SNR)~\cite{Bcal06,Dal08,Bal08,ZL09}. Particularly in \cite{Bcal06}, the end-to-end outage probability analysis of opportunistic relaying without direct link between source and destination nodes was presented. In addition, several works have considered the OR scheme under DF protocol in Rayleigh fading environment, where only the upper bound for the statistics of the best relay local SNR\footnote{The statistic refers to the probability density function (PDF) of the received SNR at the destination, called $\gamma_{r_*d}$, from the best relay $r_*$.} was obtained~\cite{GDC08,Adal08}. Moreover, performance analysis of single relay selection for DF protocols were proposed in~\cite{MK08,FU09,NB09,Cal10}. In~\cite{MK08}, Michalopoulos and Karagiannidis proposed closed-form expressions for the outage and bit error probability (BEP). However, the activated relay is selected from a decoding set, so that the input signal-to-noise ratio (SNR) is compared to a threshold before forwarding, and the diversity order was not derived explicitly. In~\cite{FU09}, Fareed and Uysal considered a relay selection method in a DF multi-relay network where the selected relay cooperates only if the SNR of the source-destination (direct) link is less than the minimum of the channel gains of the first and the second hops. The authors proposed an approximated closed-form symbol error rate (SER) expression. Recently, Nikjah and Beaulieu in~\cite{NB09} offered the first exact performance analysis of opportunistic DF relaying. However~\cite{NB09} focused on outage probability and ergodic capacity performance metrics and the end results were expressed in integral forms. However, in~\cite{Cal10} , Chen~\textit{\textit{e}t. al} derived only approximate symbol error probability (SEP) expression in integral form for opportunistic DF relaying.
\subsection{Contributions of this Paper}
\label{Contri}In this paper we consider a half duplex DF-based cooperative two-hop communications where an opportunistic relaying problem is considered. We state that the objective of this paper is not to revisit path selection, but to focus on giving valid accurate analysis over all SNR regimes. In fact, we determine the exact closed-form expressions of the end-to-end bit error rate (BER) where the source may or may not be able to communicate directly with the destination due to the shadowing. In particular, we consider the important effect of the possible erroneously detected and transmitted data at the regenerative relay. Our analytical approach requires that we determine the probability density function (PDF) of the received SNR by and from the selected relay, called $\gamma_{sr_*}$ and $\gamma_{r_*d}$, respectively. To the best of our knowledge, such performance analysis based on exact statistics (explicit form) of each hop has not been considered in the literature, and using the newly derived exact statistics, we investigate the asymptotic error performance and find the diversity order of these systems.
\subsection{Organization of this Paper}
The remainder of this paper is organized as follows. In section~\ref{systmodel}, we introduce the system model and the statistics of each hop. In section~\ref{Perf}, the accurate closed form for the end-to-end BER is derived and the diversity order of each scheme is determined. Finally, the simulation results for symmetric and linear networks are depicted in section~\ref{simulations} while some concluding remarks are given in section~\ref{conclu}.
\section{System model}
\label{systmodel}In this section, we describe our proposed cooperative diversity scheme in which the source may or may not be able to communicate directly with the destination, and we note that only a selected relay from a cluster is targeted to cooperate. The source, destination, and relays are denoted as S, D and $r_k$ where $k \in \{1,...,K\}$. We assume that each terminal is equipped with one antenna. We denote $h_{sr_k}$, $h_{sd}$ and $h_{r_{k}d}$ as the coefficients of the channels between the source (S) and the $k^{th}$ relay, the source and the destination (D), and the $k^{th}$ relay and the destination, modeled as flat fading and Rayleigh distributed with variances $\sigma_{sr_k}^{2}$, $\sigma_{sd}^{2}$ and $\sigma_{r_{k}d}^{2}$, respectively.
\subsection{Fixed Selection Cooperative Relaying (FSCR)}\label{FSCR}
The source broadcasts the symbols $s(n)$ which are received by (D) and each relay $r_k$ as
\begin{subequations}\label{eqbr1}
\begin{align}
      y_{d}(n) &= \sqrt{E_s} h_{sd} s(n) + n_{d}(n) \label{eqbr1_1}\\
      y_{r_k}(n) &= \sqrt{E_s} h_{sr_k} s(n) + n_{r_k}(n) \label{eqbr1_2},
\end{align}
\end{subequations}
respectively, where $n_{d}(n)$ and $n_{r_k}(n)$ are the additive-noise symbols at the destination and the $k^{th}$ relay,
respectively, with the same variance $N_{0}$, and $E_s$ is the symbol energy. Hence, we denote $\gamma_{sd} = E_s |h_{sd}|^{2}/N_{0}$ (resp. $\gamma_{sr_k} = E_s |h_{sr_k}|^{2}/N_{0}$) and $\gamma_{r_{k}d} = E_s |h_{r_{k}d}|^{2}/N_{0}$ the instantaneous received SNR at the destination (resp. at the $k^{th}$ relay) from the source and the $k^{th}$ relay, respectively, and $\bar{\gamma}_{sd} = \sigma_{sd}^{2} E_{s}/N_0$, $\bar{\gamma}_{sr_k} = \sigma_{sr_k}^{2} E_{s}/N_0$ are the average received SNR at the destination and the $k^{th}$ relay, respectively. We assume that the relays are close to each other and forming a cluster\footnote{We assume short distances between the relays compared to the distances (S)-cluster, and cluster-(D), respectively.} and we assume that the relays and the destination receive the same average SNRs $\bar{\gamma}_{sr_k}$ and $\bar{\gamma}_{r_{k}d}$ from the source and the relays, respectively. Thus, we denote $\bar{\gamma}_{rd} = \bar{\gamma}_{r_{k}d}$ and $\bar{\gamma}_{sr} = \bar{\gamma}_{sr_k}$ for all $k$.\\
During the second hop, only a selected relay $r_*$ will transmit using the DF protocol,
\begin{equation}\label{eqd}
    y_{d}(n+1) = \sqrt{E_s} h_{r_*d} \widetilde{s}(n) + n_{d}(n+1) ,
\end{equation}
where $\widetilde{s}(n)$ is the decoded and retransmitted signal by the best relay $r_*$ which is selected following the rule
\begin{equation}
    r_* = \arg\max_k\min\left(\gamma_{sr_{k}},\gamma_{r_{k}d}\right),
\end{equation}
where $\min\left(\gamma_{sr_{k}},\gamma_{r_{k}d}\right)$ represents a bottleneck in term of end-to-end capacity\cite{Adal08}. Therefore, the destination combines the received signals from (S) and $r_*$ using a maximum ratio combining (MRC) detector as
\begin{equation}
    y_{c} = \left(h_{sd}\right)^{*} y_{d}(n) + \left(h_{rd}\right)^{*} y_{d}(n+1).
\end{equation}
Hence, the combined SNR at the destination, called $\beta$, is the sum of the two independent SNRs $\gamma_{sd}$ and $\gamma_{r_*d}$ with the corresponding PDFs $p_{\gamma_{sd}}(.)$ and $p_{\gamma_{r_*d}}(.)$ where
\begin{equation}\label{pdfSD}
    p_{\gamma_{sd}}(\mathrm{y}) = \frac{1}{\bar{\gamma}_{sd}}~e^{-\mathrm{y}/\bar{\gamma}_{sd}},
\end{equation}
and the PDF of $\gamma_{r_*d}$ may be shown to be given by (see Appendix~\ref{AppA}),
\begin{equation}\label{eqPDF}
    p_{\gamma_{r_*d}}(\mathrm{x}) = \sum_{i=1}^K \left(_i^K\right)~\frac{(-1)^{i-1}}{\bar{\gamma}_{sr}} \frac{i\bar{\gamma}}{i\bar{\gamma}_{rd}-\bar{\gamma}}\left(e^{-\mathrm{x}/\bar{\gamma}_{rd}} - e^{-i\mathrm{x}/\bar{\gamma}}\right) + \sum_{i=1}^K \left(_i^K\right)~\frac{(-1)^{i-1}}{\bar{\gamma}_{rd}}~i~e^{-i\mathrm{x}/\bar{\gamma}}
\end{equation}
where $\bar{\gamma} = \frac{\bar{\gamma}_{sr}\bar{\gamma}_{rd}}{\bar{\gamma}_{sr}+\bar{\gamma}_{rd}}$ and $i\bar{\gamma}_{rd} \neq \bar{\gamma}$ $\forall i = 1, 2, \cdots, K$. Therefore, the PDF of $\beta = \gamma_{sd} + \gamma_{r_*d}$ can be obtained by the convolution of the PDF of $\gamma_{sd}$ and $p_{\gamma_{r_*d}}$, as
\begin{equation}\label{pBB}
    p_{\beta}(\beta) = \int_{0}^\beta p_{\gamma_{sd}}(\mathrm{x})p_{\gamma_{r_*d}}(\beta-\mathrm{x}) d\mathrm{x}
\end{equation}
which is expressed in~(\ref{PDFbeta}).
\begin{figure*}[!t]
\normalsize
\begin{eqnarray}\label{PDFbeta}
    p_{\beta}(\beta) &=& \sum_{i=1}^{K} \left(_i^K\right)~\frac{(-1)^{i-1}}{\bar{\gamma}_{sr}}\frac{i\bar{\gamma}}{i\bar{\gamma}_{rd}-\bar{\gamma}} \left[\frac{\bar{\gamma}_{rd}}{\bar{\gamma}_{sd}-\bar{\gamma}_{rd}}\left(e^{-\beta/\bar{\gamma}_{sd}}-e^{-\beta/\bar{\gamma}_{rd}}\right) - \frac{\bar{\gamma}}{i\bar{\gamma}_{sd}-\bar{\gamma}} \left(e^{-\beta/\bar{\gamma}_{sd}}-e^{-i\beta/\bar{\gamma}}\right)\right] \nonumber \\
    &{+}& \sum_{i=1}^{K} \left(_i^K\right)~\frac{(-1)^{i-1}}{\bar{\gamma}_{rd}}\frac{i\bar{\gamma}}{i\bar{\gamma}_{sd}-\bar{\gamma}}\left(e^{-\beta/\bar{\gamma}_{sd}} - e^{-i\beta/\bar{\gamma}}\right)
\end{eqnarray}
\hrulefill
\vspace*{-4pt}
\end{figure*}
\subsection{Distributed Selection Combining (DSC) Scheme}\label{DSC}
In this scheme, the destination chooses whether to receive from the direct link S-D or the relayed branch according to the instantaneous SNRs $\gamma_{sd}$ and $\gamma_{r_*d}$, respectively. Otherwise, the instantaneous SNR at the output of the selection combining (SC) detector is given by:
\begin{equation}\label{eq1SC}
    \gamma_{DSC} = \max\left(\gamma_{sd},\gamma_{r_*d}\right).
\end{equation}
It can be noted that the statistics of $\gamma_{DSC}$ depends on statistics of $\gamma_{sd}$ and $\gamma_{r_*d}$. In particular, the cumulative density function (CDF) of $\gamma_{DSC}$ is given by
\begin{equation}\label{cdfDSC}
    F_{\gamma_{DSC}}(\mathrm{x}) = F_{\gamma_{sd}}(\mathrm{x})~F_{\gamma_{r_*d}}(\mathrm{x}),
\end{equation}
where the CDFs of $\gamma_{sd}$ and $\gamma_{r_*d}$, denoted as $F_{\gamma_{sd}}(.)$ and $F_{\gamma_{r_*d}}(.)$, respectively, can be derived using the PDFs of $\gamma_{sd}$ and $\gamma_{r_*d}$ in (\ref{pdfSD}) and (\ref{eqPDF}), respectively, as
\begin{equation}\label{eqCDFsd}
    F_{\gamma_{sd}}(\mathrm{x}) = 1-e^{-\mathrm{x}/\bar{\gamma}_{sd}},
\end{equation}
and
\begin{equation}\label{eqCDFrd}
    F_{\gamma_{r_*d}}(\mathrm{x}) = \sum_{i=1}^K \left(_i^K\right)~\frac{(-1)^{i-1}}{\bar{\gamma}_{sr}}~\frac{i\bar{\gamma}}{i\bar{\gamma}_{rd}-\bar{\gamma}}\left[\Psi_{\frac{1}{\bar{\gamma}_{rd}}}(\mathrm{x}) - \Psi_{\frac{i}{\bar{\gamma}}}(\mathrm{x})\right] + \sum_{i=1}^K \left(_i^K\right)~\frac{(-1)^{i-1}i}{\bar{\gamma}_{rd}}~\Psi_{\frac{i}{\bar{\gamma}}}(\mathrm{x}),
\end{equation}
where
\begin{equation}\label{eq5SCR}
    \Psi_a(\mathrm{x}) = \frac{1}{a}~\left(1-e^{-a\mathrm{x}}\right).
\end{equation}
\subsection{Selection Relaying (SR) Scheme}\label{SR}
In this scheme, it is assumed that the direct link is in deep fading. Hence only $r_*$ will receive the information reliably from the source during the first phase given by~(\ref{eqbr1_2}), and the destination decodes only the message coming from $r_*$ as given by~(\ref{eqd}).
\section{BER analysis}
\label{Perf}
\subsection{Fixed Selection Cooperative Relaying}\label{SCRan}
The destination combines the received signals such as the relay can retransmit an erroneously decoded message. The end-to-end probability of error can be expressed as
\begin{equation}\label{eqSCR1}
    P_{e,FSCR} = P_{prop} P_{sr_*} + \left(1 - P_{sr_*}\right) P_{mrc}
\end{equation}
where $P_{prop}$ denotes the error propagation probability which can be tightly approximated for a BPSK modulation by
\begin{equation}\label{eqSCR2}
    P_{prop} \approx \frac{\bar{\gamma}_{r_*d}}{\bar{\gamma}_{r_*d} + \bar{\gamma}_{sd}},
\end{equation}
where $\bar{\gamma}_{r_*d}$, the expected value of $\gamma_{r_*d}$, can be easily verified to be expressed as
\begin{equation}\label{eqbar}
    \bar{\gamma}_{r_*d} = \sum_{i=1}^K \left(_i^K\right)~\frac{(-1)^{i-1}}{\bar{\gamma}_{sr}}\frac{i\bar{\gamma}}{i\bar{\gamma}_{rd}-\bar{\gamma}}\left[\bar{\gamma}_{rd}^2 - \left(\frac{\bar{\gamma}}{i}\right)^2\right] + \sum_{i=1}^K \left(_i^K\right)~\frac{(-1)^{i-1}}{i\bar{\gamma}_{rd}}~\bar{\gamma}^2 ,
\end{equation}
and $P_{sr_*}$ is the probability of error for the communication link between the source and the relay which is derived by
\begin{equation}\label{eqSCR3}
    P_{sr_*} = \int_0^{\infty} \frac{1}{2}~\mathrm{erfc}(\sqrt{\mathrm{x}})p_{\gamma_{sr_*}}(\mathrm{x})\mathrm{dx}
\end{equation}
where $p_{\gamma_{sr_*}}(.)$ can be derived as in~(\ref{eqPDF}) by replacing $\gamma_{r_*d}$ with $\gamma_{sr_*}$ and $\mathrm{erfc}(.)$ is the complementary error function. Hence performing the integration in (\ref{eqSCR3}), and using the following identity
\begin{equation}
  l(\alpha) \triangleq \frac{1}{2}~\int_0^\infty \mathrm{erfc}(\sqrt{\beta}) e^{-\alpha\beta} \mathrm{d}\beta = \frac{1}{2\alpha}\left[1-\frac{1}{\sqrt{1+\alpha}}\right],
\end{equation}
$P_{sr_*}$ is found to be expressed as in~(\ref{eqSCR4}).
\begin{figure*}[!t]
\normalsize
\begin{equation}\label{eqSCR4}
    P_{sr_*} = \sum_{i=1}^{K} \left(_i^K\right)\frac{(-1)^{i-1}}{\bar{\gamma}_{rd}} \frac{i\bar{\gamma}}{i\bar{\gamma}_{sr}-\bar{\gamma}}\left[l\left(\frac{1}{\bar{\gamma}_{sr}}\right) - l\left(\frac{i}{\bar{\gamma}}\right)\right] + \sum_{i=1}^{K} \left(_i^K\right)~\frac{i (-1)^{i-1}}{\bar{\gamma}_{sr}}l\left(\frac{i}{\bar{\gamma}}\right).
\end{equation}
\hrulefill
\vspace*{-6pt}
\end{figure*}
In~(\ref{eqSCR1}), we need also $P_{mrc}$ which is the error probability of the combined direct and opportunistic paths, given by
\begin{equation}\label{MRCint}
    P_{mrc} = \int_{0}^{\infty} \frac{1}{2}~\mathrm{erfc}(\sqrt{\beta})p(\beta) \mathrm{d}\beta ,
\end{equation}
which is given in~(\ref{eqSCR5}).
\begin{figure*}[!t]
\normalsize
\setlength{\arraycolsep}{0.0em}
\begin{eqnarray}\label{eqSCR5}
    P_{mrc} &=& \sum_{i=1}^{K} \left(_i^K\right)\frac{(-1)^{i-1}}{\bar{\gamma}_{sr}}\frac{i\bar{\gamma}}{i\bar{\gamma}_{rd}-\bar{\gamma}} \left[\frac{\bar{\gamma}_{rd}}{\bar{\gamma}_{sd}-\bar{\gamma}_{rd}} \left(l\left(\frac{1}{\bar{\gamma}_{sd}}\right)-l\left(\frac{1}{\bar{\gamma}_{rd}}\right)\right)
    -\frac{\bar{\gamma}}{i\bar{\gamma}_{sd}-\bar{\gamma}}\left(l\left(\frac{1}{\bar{\gamma}_{sd}}\right)-l\left(\frac{i}{\bar{\gamma}}\right)\right)\right] \nonumber \\
    &&{+} \sum_{i=1}^{K} \left(_i^K\right)\frac{(-1)^{i-1}}{\bar{\gamma}_{rd}}\frac{i\bar{\gamma}}{i\bar{\gamma}_{sd}-\bar{\gamma}}\left(l\left(\frac{1}{\bar{\gamma}_{sd}}\right)-l\left(\frac{i}{\bar{\gamma}}\right)\right)
\end{eqnarray}
\setlength{\arraycolsep}{2pt}
\hrulefill
\vspace*{-4pt}
\end{figure*}
Finally, with (\ref{eqSCR2}), (\ref{eqbar}), (\ref{eqSCR4}) and (\ref{eqSCR5}), the end-to-end probability of error can be easily evaluated.
\begin{lemma}\label{lemma1} For a source-destination pair with $K$ potential relays in Rayleigh fading channels, the end-to-end BER of the fixed selection cooperative relaying scheme in the high-SNR regime, is
\begin{equation}\label{eqLemma1}
    P_{e,FSCR} \approx P_{prop}\frac{\Gamma(K+\frac{1}{2})}{2\sqrt{\pi}}\frac{1}{\bar{\gamma}_{sr}}\left(\frac{1}{\bar{\gamma}}\right)^{K-1} + \frac{\Gamma(K+\frac{3}{2})}{2\sqrt{\pi}(K+1)}\frac{1}{\bar{\gamma}_{sd}}\frac{1}{\bar{\gamma}_{rd}}\left(\frac{1}{\bar{\gamma}}\right)^{K-1},
\end{equation}
\end{lemma}
\begin{proof}See Appendix~\ref{AppB}.\end{proof}
\subsection{Distributed Selection Combining Scheme}
The destination selects the best coming path and the end-to-end BER is found to be
\begin{equation}\label{eqDSC1}
    P_{e,DSC} = P_{prop} P_{sr_*} + \left(1 - P_{sr_*}\right) P_{DSC},
\end{equation}
where $P_{prop}$ and $P_{sr_*}$ are detailed above and $P_{DSC}$ is the probability of error for the selected link communication to the destination. Based on a general result in~\cite[Eq. 32]{T04}, we can derive $P_{DSC}$ for a BPSK modulation as
\begin{equation}\label{eqDSC2}
    P_{DSC} = \frac{1}{2\sqrt{\pi}}\int_{0}^{\infty} \frac{e^{-\mathrm{z}}}{\sqrt{\mathrm{z}}}~F_{\gamma_{DSC}}(\mathrm{z})d\mathrm{z},
\end{equation}
which can be rewritten as
\begin{equation}\label{eqDSC3}
    P_{DSC} = \underbrace{\frac{1}{2\sqrt{\pi}}\int_{0}^{\infty} \frac{e^{-\mathrm{z}}}{\sqrt{\mathrm{z}}}~F_{\gamma_{r_*d}}(\mathrm{z})d\mathrm{z}}_{I_1} - \underbrace{\frac{1}{2\sqrt{\pi}}\int_{0}^{\infty} \frac{e^{-\mathrm{z}(1+1/\bar{\gamma}_{sd})}}{\sqrt{\mathrm{z}}}~F_{\gamma_{r_*d}}(\mathrm{z})d\mathrm{z}}_{I_2}.
\end{equation}
It should be noted that $I_1$ defines the probability of error for the communication link between $r_*$ and D, which can be expressed as in~(\ref{eqDSC4})
\begin{figure*}[!t]
\normalsize
\begin{equation}\label{eqDSC4}
    I_1 \triangleq P_{r_*d} = \sum_{i=1}^{K} \left(_i^K\right)\frac{(-1)^{i-1}}{\bar{\gamma}_{sr}} \frac{i\bar{\gamma}}{i\bar{\gamma}_{rd}-\bar{\gamma}}\left[l\left(\frac{1}{\bar{\gamma}_{rd}}\right) - l\left(\frac{i}{\bar{\gamma}}\right)\right] + \sum_{i=1}^{K} \left(_i^K\right) \frac{i(-1)^{i-1}}{\bar{\gamma}_{rd}}~l\left(\frac{i}{\bar{\gamma}}\right),
\end{equation}
\hrulefill
\vspace*{-6pt}
\end{figure*}
and $I_2$ can be derived with the help of the following identity
\begin{equation}\label{eqDSC5}
    \Theta_{a} \triangleq \frac{1}{2\sqrt{\pi}} \int_{0}^{\infty} \frac{e^{-\mathrm{z}(1+1/\bar{\gamma}_{sd})}}{\sqrt{\mathrm{z}}}~\Psi_{a}(\mathrm{z})d\mathrm{z} = \frac{1}{2a}\left[\sqrt{\frac{1}{1 + \frac{1}{\bar{\gamma}_{sd}}}} - \sqrt{\frac{1}{1 + a + \frac{1}{\bar{\gamma}_{sd}}}}~\right].
\end{equation}
and be given by~(\ref{eqDSC6}).\\
\begin{figure*}[!t]
\normalsize
\begin{equation}\label{eqDSC6}
    I_{2} = \sum_{i=1}^{K} \left(_i^K\right)\frac{(-1)^{i-1}}{\bar{\gamma}_{sr}} \frac{i\bar{\gamma}}{i\bar{\gamma}_{rd}-\bar{\gamma}}\left[\Theta_{\frac{1}{\bar{\gamma}_{rd}}} - \Theta_{\frac{i}{\bar{\gamma}}}\right] + \sum_{i=1}^{K} \left(_i^K\right) \frac{i(-1)^{i-1}}{\bar{\gamma}_{rd}}~\Theta_{\frac{i}{\bar{\gamma}}},
\end{equation}
\hrulefill
\vspace*{-6pt}
\end{figure*}
Finally, with (\ref{eqSCR2}), (\ref{eqbar}), (\ref{eqSCR4}), (\ref{eqDSC4}) and (\ref{eqDSC6}), the end-to-end probability of error can be easily evaluated.
\begin{lemma}\label{lemma2} For a source-destination pair with $K$ potential relays in Rayleigh fading channels, the end-to-end BER of the distributed selection combining scheme in the high-SNR regime, is
\begin{equation}\label{eqLemma2}
    P_{e,DSC} \approx P_{prop}\frac{\Gamma(K+\frac{1}{2})}{2\sqrt{\pi}}\frac{1}{\bar{\gamma}_{sr}}\left(\frac{1}{\bar{\gamma}}\right)^{K-1} + \frac{\Gamma(K+\frac{3}{2})}{2\sqrt{\pi}}\frac{1}{\bar{\gamma}_{sd}}\frac{1}{\bar{\gamma}_{rd}}\left(\frac{1}{\bar{\gamma}}\right)^{K-1}.
\end{equation}
\end{lemma}
\begin{proof}See Appendix~\ref{AppC}.\end{proof}
\subsection{Selection Relaying}\label{SRan}
In this form of relaying it is assumed that the direct path is unusable due to the deep fade instances or heavy shadowing~\cite{Adal08}. Therefore the end-to-end BER is found to be
\begin{equation}\label{eqSR1}
    P_{e,SR} = P_{sr_*} + P_{r_*d} - P_{sr_*}P_{r_*d},
\end{equation}
where $P_{sr_*}$ was already defined by~(\ref{eqSCR4}), and $P_{r_*d}$ is the probability of error for the communication link between $r_*$ and D which can be expressed by~(\ref{eqDSC4}).
\begin{lemma}\label{lemma3} For a source-destination pair with $K$ potential relays in Rayleigh fading channels, the end-to-end BER of the selection relaying scheme in the high-SNR regime, is
\begin{equation}\label{eqLemma3}
    P_{e,SR} \approx \frac{\Gamma(K+\frac{1}{2})}{2\sqrt{\pi}}~\left(\frac{1}{\bar{\gamma}}\right)^{K}.
\end{equation}
\end{lemma}
\section{Performance results}\label{simulations}
\subsection{Network Geometry}
We anticipate that cooperation will perform differently as function
of the positions of the mobiles with respect to the destination. Hence
we study an {\em asymmetric} or {\em linear} network (LN) where we model the path-loss, i.e. the mean channel powers $\sigma_{ij}^2$,
as a function of the relays cluster position $d$ by
\begin{equation}\label{eqsigma}
  \sigma_{sd}^2 = 1,~ \sigma_{sr_k}^2 = d^{-\nu},~ \sigma_{r_{k}d}^2 = (1-d)^{-\nu},
\end{equation}
where $\nu$ is the path loss exponent and $0 < d (= \textrm{distance}_{s-cluster}) < 1$. The
distances are normalized by the distance $d_{sd}$. In these
coordinates, the source can be located at (0,0), the destination can
be located at (1,0), without loss of generality, and the relays are
located at ($d$,0).
\subsection{Simulation Results}
\label{SimRes}In this section, we evaluate the performance of our schemes in terms of the end-to-end BER at the destination as function of the SNR $= E_b/N_0$ for a number ($K$) of potential relays in phase II. All schemes were simulated assuming BPSK modulation. It is also assumed that the amplitudes of the fading from each transmit antenna to each receive antenna are uncorrelated in the case of cooperative selection relaying scheme and Rayleigh distributed. Furthermore, we assumed that all receivers have the same noise properties. This implies that in all depicted figures, the noise power of all paths is the same. Further, we assumed that the receiver has perfect knowledge of the channels.\\
Figures~\ref{fig:fig1}-\ref{fig:fig3} depict the end-to-end error-rate performance and the corresponding asymptotic curves in LN networks as function of SNR for FSCR, DSC and SR schemes, respectively, where a relays cluster is located at different distances $d$ from the source. Figure~\ref{fig:fig4} shows performance comparison between FSCR and DSC schemes. All figures compare the analytical and simulation results for $K = 2$ and $K = 4$, respectively.\\
Figures~\ref{fig:fig1} and \ref{fig:fig2} depict the end-to-end BER as function of the SNR for the FSCR and DSC schemes, respectively, where the relays cluster is located at $d = 0.1$ and $d = 0.5$, respectively. It could be noted that the diversity order is $K$ when the relay cluster is located at the mid-distance to the destination. This is due to the fact that the error propagation probability, $P_{prop}$, is becoming close to 1. The full diversity order is recovered when $d = 0.1$. In addition, It may be noted that the gap, between FSCR and DSC cures, is shrinked when $d = 0.5$. Therefore, DSC scheme could be considered as appropriate since its BER penalty is minor as shown in figure~\ref{fig:fig4}, and it is considered as the less complicated than MRC~\cite{AlouiniBook}.\\
Figure~\ref{fig:fig3} depicts the end-to-end BER as function of the SNR for the SR scheme for the same network architectures. We note that the SR scheme do better when the relay cluster is located in the middle between the the source and the destination, and the full diversity order is achieved as expected by~(\ref{eqLemma3}). Our proposed analysis is well confirmed by the simulation results.
\section{Conclusions}
\label{conclu}In this work, we studied three opportunistic cooperation protocols namely fixed selection cooperative relaying, distributed selection combining and selection relaying, based on DF transmission in a Rayleigh fading environment. We provided exact statistics, and as result, we presented the BER performance analysis as well as the asymptotic analysis. We performed several simulations to confirm our theoretical analysis.
\appendix
\subsection{Derivation of Eq.~(\ref{eqPDF})}\label{AppA}
Based on our problem formulation we can write the PDF $p_{\gamma_{r_*d}}(\mathrm{x})$ as follows
\begin{equation}\label{A1eq1}
    p_{\gamma_{r_*d}}(\mathrm{x}) = \int_0^{\infty} p_{\gamma_{r_id}/Z_i = \mathrm{z}}(\mathrm{x}) . p_{\max(Z_i)}(\mathrm{z}) d\mathrm{z},
\end{equation}
where $Z_i = \min(\gamma_{sr_i},\gamma_{r_id})$. \\
Using the Bayes rule, it is well known that the conditional probability density function can be expressed as
\begin{equation}\label{A1eq2}
    p_{\gamma_{r_id}/Z_i = \mathrm{z}}(\mathrm{x}) = \frac{p_{\gamma_{r_id},Z_i}(\mathrm{x},\mathrm{z})}{p_{Z_i}(\mathrm{z})}.
\end{equation}
Now, we can show that the cumulative density function (CDF) of $Z_i$ can be expressed as
\begin{equation}\label{A1eq3}
    F_{Z_i}(\mathrm{z}) = 1 - Pr[\gamma_{sr_i} \geq \mathrm{z}]Pr[\gamma_{r_id} \geq \mathrm{z}] = 1 - e^{-\mathrm{z}/\bar{\gamma}},
\end{equation}
where $\bar{\gamma} = \frac{\bar{\gamma}_{sr}\bar{\gamma}_{rd}}{\bar{\gamma}_{sr}+\bar{\gamma}_{rd}}$, and the joint CDF $F_{\gamma_{r_id},Z_i}(\mathrm{x},\mathrm{z})$ can be expressed as
\begin{eqnarray}\label{A1eq4}
   &&{} F_{\gamma_{r_id},Z_i}(\mathrm{x},\mathrm{z}) = Pr[\gamma_{r_id}<\mathrm{x},\min(\gamma_{sr_i},\gamma_{r_id})<\mathrm{z}] \nonumber \\
                                                 &&{=} \left\{\begin{array}{ll}
                                                              F_{\gamma_{r_id}}(\mathrm{x})-\left(F_{\gamma_{r_id}}(\mathrm{x})-F_{\gamma_{r_id}}(\mathrm{z})\right)\left(1-F_{\gamma_{sr_i}}(\mathrm{z})\right), & \mathrm{x} \geq \mathrm{z} \\
                                                              F_{\gamma_{r_id}}(\mathrm{x}), & \mathrm{x} < \mathrm{z}
                                                            \end{array}
                                                  \right.
\end{eqnarray}
We note that the joint CDF $F_{\gamma_{r_id},Z_i}(\mathrm{x},\mathrm{z})$ is not continuous along the $\mathrm{x}$ direction at $\mathrm{x} = \mathrm{z}$. Therefore, the result of the derivative (i.e the joint PDF $p_{\gamma_{r_id},Z_i}(\mathrm{x},\mathrm{z})$) involves an impulse at the position $\mathrm{x} = \mathrm{z}$. Specially, the joint PDF $p_{\gamma_{r_id},Z_i}(\mathrm{x},\mathrm{z})$ is given by
\begin{equation}
   p_{\gamma_{r_id},Z_i}(\mathrm{x},\mathrm{z}) = \left\{\begin{array}{ll}
                                                              p_{\gamma_{r_id}}(\mathrm{x})p_{\gamma_{sr_i}}(\mathrm{z}) + p_{\gamma_{r_id}}(\mathrm{x})\left(1-F_{\gamma_{sr_i}}(\mathrm{z})\right)\delta(\mathrm{x}-\mathrm{z}) , & \mathrm{x} \geq \mathrm{z} \\
                                                              0, & \mathrm{x} < \mathrm{z}
                                                            \end{array}
                                                  \right.
\end{equation}
It follows the PDF of $\gamma_{r_*d}$ is given by
\begin{equation}\label{A1eq6}
    p_{\gamma_{r_*d}}(\mathrm{x}) = \int_0^{\mathrm{x}} \frac{p_{\gamma_{r_id}}(\mathrm{x})p_{\gamma_{sr_i}}(\mathrm{z})}{p_{Z_i}(\mathrm{z})}~p_{\max(Z_i)}(\mathrm{z}) d\mathrm{z} + \frac{p_{\gamma_{r_id}}(\mathrm{x}) \left(1-F_{\gamma_{sr_i}}(\mathrm{x})\right)}{p_{Z_i}(\mathrm{x})}~p_{\max(Z_i)}(\mathrm{x})
\end{equation}
where
\begin{equation}\label{A1eq7}
    p_{Z_i}(\mathrm{z}) = \frac{1}{\bar{\gamma}} e^{-\mathrm{z}/\bar{\gamma}},
\end{equation}
It can be easily shown that $F_{\max(Z_i)}(\mathrm{z})$ be expressed as
\begin{equation}\label{A1eq8}
    F_{\max(Z_i)}(\mathrm{z}) = \prod_{i=1}^{K} Pr\left[Z_i < \mathrm{z}\right] = \left(1 - e^{-\mathrm{z}/\bar{\gamma}}\right)^K
\end{equation}
Therefore after taking the derivative, we have
\begin{equation}\label{A1eq9}
    p_{\max(Z_i)}(\mathrm{z}) = \frac{K}{\bar{\gamma}}e^{-\mathrm{z}/\bar{\gamma}}\left(1 - e^{-\mathrm{z}/\bar{\gamma}}\right)^{K-1} = \sum_{i=1}^{K} \left(_i^K\right)~(-1)^{i-1}~\frac{i}{\bar{\gamma}}~e^{-i\mathrm{z}/\bar{\gamma}}.
\end{equation}
where the second equality holds from the binomial expansion.\\
Substituting (\ref{A1eq7}) and (\ref{A1eq9}) in (\ref{A1eq1}), Eq.~(\ref{eqPDF}) is derived by performing the integration.
\subsection{Derivation of Eq.~(\ref{eqLemma1})}\label{AppB}
It is easy to note that $P_{e,FSCR}$ could be approximated by
\begin{equation}\label{Beq1}
    P_{e,FSCR}^{\infty} = P_{prop} P_{sr_*}^{\infty} + P_{mrc}^{\infty}
\end{equation}
where $P_{sr_*}^{\infty}$ and $P_{mrc}^{\infty}$ are the approximated expressions of $P_{sr_*}$ and $P_{mrc}$, respectively.\\
To this end, we needed the approximated expressions of the PDFS $p_{\gamma_{sr_*}}(.)$ and $p_{\beta}(.)$ which are given by
\begin{equation}\label{Beq2}
    p_{\gamma_{sr_*}}(\mathrm{x}) \approx \frac{K}{\bar{\gamma}_{sr}} \left(\frac{1}{\bar{\gamma}}\right)^{K-1} \mathrm{y}^{K-1}, ~\mathrm{y} > 0,
\end{equation}
and
\begin{equation}\label{Beq3}
    p_{\beta}(\beta) \approx \frac{1}{\gamma_{sd}}\frac{1}{\bar{\gamma}_{rd}}\left(\frac{1}{\bar{\gamma}}\right)^{K-1} \beta^{K}, ~\mathrm{\beta} > 0
\end{equation}
Based on (\ref{Beq2}) and (\ref{Beq3}), it becomes easy to derive $P_{sr_*}^{\infty}$ and $P_{mrc}^{\infty}$ by using integrations in (\ref{eqSCR3}) and (\ref{MRCint}).
\subsection{Derivation of Eq.~(\ref{eqLemma2})}\label{AppC}
It is easy to note that $P_{e,DSC}$ could be approximated by
\begin{equation}\label{Ceq1}
    P_{e,DSC}^{\infty} = P_{prop} P_{sr_*}^{\infty} + P_{DSC}^{\infty}
\end{equation}
where $P_prop$ and $P_{sr_*}^{\infty}$ are already given, and $P_{DSC}^{\infty}$ is the approximated expressions of $P_{DSC}$ which can be derived by integrating the approximated expression of $F_{\gamma_{DSC}}(.)$. For this end, let us start by approximating $F_{\gamma_{sd}}(.)$ and $F_{\gamma_{r_*d}}(.)$, given by
\begin{equation}\label{Ceq2}
    F_{\gamma_{sd}} \approx \frac{\mathrm{z}}{\bar{\gamma}_{sd}},~\mathrm{z} > 0,
\end{equation}
and
\begin{equation}\label{Ceq3}
    F_{\gamma_{r_*d}}(\mathrm{z}) \approx \frac{1}{\bar{\gamma}_{rd}}~\frac{1}{\bar{\gamma}}~\mathrm{z}^{K},~\mathrm{z} > 0.
\end{equation}
Using the following identity
\begin{equation}\label{Ceq4}
    \int_0^{\infty} \mathrm{z}^{K+\frac{1}{2}} e^{-\mathrm{z}} d\mathrm{z} = \Gamma\left(K + \frac{3}{2}\right),
\end{equation}
and substituting (\ref{Ceq2}) and (\ref{Ceq3}) in (\ref{eqDSC2}), Eq.~(\ref{eqLemma2}) is derived by performing the integration.
\bibliographystyle{IEEEtran}
\bibliography{IEEEabrv,refJrlTYA}

\begin{thebibliography}{10}
\providecommand{\url}[1]{#1}
\csname url@rmstyle\endcsname
\providecommand{\newblock}{\relax}
\providecommand{\bibinfo}[2]{#2}
\providecommand\BIBentrySTDinterwordspacing{\spaceskip=0pt\relax}
\providecommand\BIBentryALTinterwordstretchfactor{4}
\providecommand\BIBentryALTinterwordspacing{\spaceskip=\fontdimen2\font plus
\BIBentryALTinterwordstretchfactor\fontdimen3\font minus
  \fontdimen4\font\relax}
\providecommand\BIBforeignlanguage[2]{{%
\expandafter\ifx\csname l@#1\endcsname\relax
\typeout{** WARNING: IEEEtran.bst: No hyphenation pattern has been}%
\typeout{** loaded for the language `#1'. Using the pattern for}%
\typeout{** the default language instead.}%
\else
\language=\csname l@#1\endcsname
\fi
#2}}

\bibitem{LaW}
J.~N. Laneman and G.~W. Wornell, ``Distributed space-time-coded protocols for
  exploiting cooperative diversity in wireless networks,'' \emph{IEEE Trans.
  Inform. Theory}, vol.~49, no.~10, pp. 2415--2425, Oct. 2003.

\bibitem{SEA1}
A.~Sendonaris, E.~Erkip, and B.~Aazhang, ``User cooperation diversity,
  \textsc{P}art {I} : System description,'' \emph{IEEE Trans. Comm.}, vol.~51,
  no.~11, pp. 1927--1938, Nov. 2003.

\bibitem{SEA2}
------, ``User cooperation diversity, \textsc{P}art {II} : Implementation
  aspects and performance analysis,'' \emph{IEEE Trans. Comm.}, vol.~51,
  no.~11, pp. 1939--1948, Nov. 2003.

\bibitem{Bal06}
A.~Bletsas, A.~Khisti, D.~P. Reed, and A.~Lippman, ``A simple cooperative
  diversity method based on network path selection,'' \emph{IEEE Journal on
  Selected Areas in Communications}, vol.~24, no.~3, pp. 659--672, Mar. 2006.

\bibitem{TAH09}
K.~Tourki, M.~S. Alouini, and M.~O. Hasna, ``Wireless transmission using
  cooperation on demand,'' \emph{Elsevier Physical Communication (2009)}, 2009,
  doi:10.1016/j.phycom.2009.10.003.

\bibitem{Kal09}
I.~Krikidis, J.~S. Thompson, S.~McLaughlin, and N.~Goertz, ``Max-min relay
  selection for legacy amplify-and-forward systems with interference,''
  \emph{IEEE Transactions on Wireless Communications}, vol.~8, no.~6, pp.
  3016--3027, Jun. 2009.

\bibitem{Bal10}
A.~Bletsas, A.~G. Dimitriou, and J.~N. Sahalos, ``Interference-limited
  opportunistic relay with reactive sensing,'' \emph{IEEE Transactions on
  Wireless Communications}, vol.~9, no.~1, pp. 14--20, Jan. 2010.

\bibitem{IA09}
S.~S. Ikki and M.~H. Ahmed, ``Exact error probability and channel capacity of
  the best-relay cooperative-diversity networks,'' \emph{IEEE Signal Processing
  Letters}, vol.~16, no.~12, pp. 1051--1054, Dec. 2009.

\bibitem{Bcal06}
A.~Bletsas, H.~Shin, M.~Z. Win, and A.~Lippman, ``Cooperative diversity with
  opportunistic relaying,'' in \emph{IEEE Wireless Communications and
  Networking Conference}, Las Vegas, USA, Apr. 2006.

\bibitem{Dal08}
Z.~Ding, Y.~Gong, T.~Ratnarajah, and C.~F.~N. Cowan, ``On the design of
  opportunistic cooperative transmission strategies with partial \textsc{CSI}
  information,'' in \emph{IEEE Vehicular Technology Conference (VTC'Spring
  08)}, Marina Bay, Singapore, May 2008.

\bibitem{Bal08}
A.~Bletsas, A.~Khisti, and M.~Z. Win, ``Opportunistic cooperative diversity
  with feedback and cheap radios,'' \emph{IEEE Transactions on Wireles
  Communications}, vol.~7, no.~5, pp. 1823--1827, May 2008.

\bibitem{ZL09}
Q.~F. Zhou, F.~C.~M. Lau, and S.~F. Hau, ``Asymptotic analysis of opportunistic
  relaying protocols,'' \emph{IEEE Transactions on Wireles Communications},
  vol.~8, no.~8, pp. 3915--3920, Aug. 2009.

\bibitem{GDC08}
B.~Gui, L.~Dai, and L.~J. Cimini, ``Selective relaying in cooperative {OFDM}
  systems: Two-hop random network,'' in \emph{Proc. IEEE Wireless comm. and
  Networking Conf. (WCNC)}, Las Vegas, USA, Apr. 2008.

\bibitem{Adal08}
A.~Adinoyi, Y.~Fan, H.~Yanikomeroglu, and H.~V. Poor, ``On the performance of
  selection relaying,'' in \emph{Proc. of the 68th IEEE Vehicular Technology
  Conference (VTC'Fall 08)}, Calgary, Canada, Sep. 2008.

\bibitem{MK08}
D.~S. Michalopoulos and G.~K. Karagiannidis, ``Performance analysis of single
  relay selection in rayleigh fading,'' \emph{IEEE Transactions on Wireless
  Communications}, vol.~7, no.~10, pp. 3718--3724, Oct. 2008.

\bibitem{FU09}
M.~M. Fareed and M.~Uysal, ``On relay selection for decode-and-forward
  relaying,'' \emph{IEEE Transactions on Wireless Communications}, vol.~8,
  no.~7, pp. 3341--3345, Jul. 2009.

\bibitem{NB09}
R.~Nikjah and N.~C. Beaulieu, ``Exact closed-form expressions for the outage
  probability and ergodic capacity of decode-and-forward opportunistic
  relaying,'' in \emph{Proc. IEEE Global Communications Conference (Globecom)},
  Honolulu, Hawaii, USA, Dec. 2009.

\bibitem{Cal10}
H.~Chen, J.~Liu, L.~Zheng, C.~Zhai, and Y.~Zhou, ``Approximate \textsc{SEP}
  analysis for \textsc{DF} cooperative networks with opportunistic relaying,''
  \emph{IEEE Signal Processing Letters}, vol.~17, no.~9, pp. 779--782, Sept.
  2010.

\bibitem{T04}
Y.~Chen and C.~Tellambura, ``Distributed functions of selection combiner output
  in equally correlated rayleigh, rician, and \texttt{N}akagami-\textit{m}
  fading channels,'' \emph{IEEE Transactions on Communications}, vol.~52,
  no.~11, pp. 1948--1956, Nov. 2004.

\bibitem{AlouiniBook}
M.~K. Simon and M.~S. Alouini, \emph{Digital Communication over Fading
  Channels}, S.~E. John~G.~Proakis, Ed.\hskip 1em plus 0.5em minus 0.4em\relax
  Wiley Series in Telecommunications and Signal Processing, 2005.

\end{thebibliography}
\begin{figure}[htbp]
\begin{center}
\includegraphics[scale=0.75]{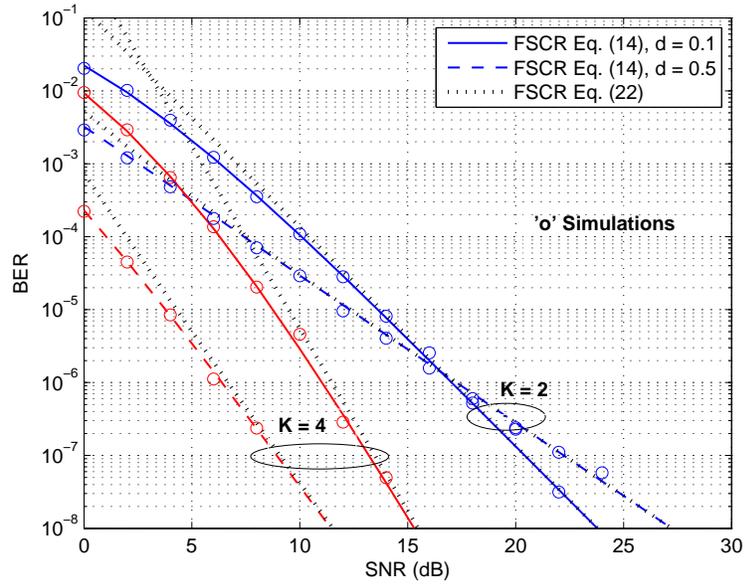}
\vspace*{-0.5 cm} \caption{End-to-end bit error rate versus SNR of the FSCR scheme using a DF transmission in the LN case for $K = 2$ and $K = 4$.}\label{fig:fig1}
\end{center}
\end{figure}
\begin{figure}[htbp]
\begin{center}
\includegraphics[scale=0.75]{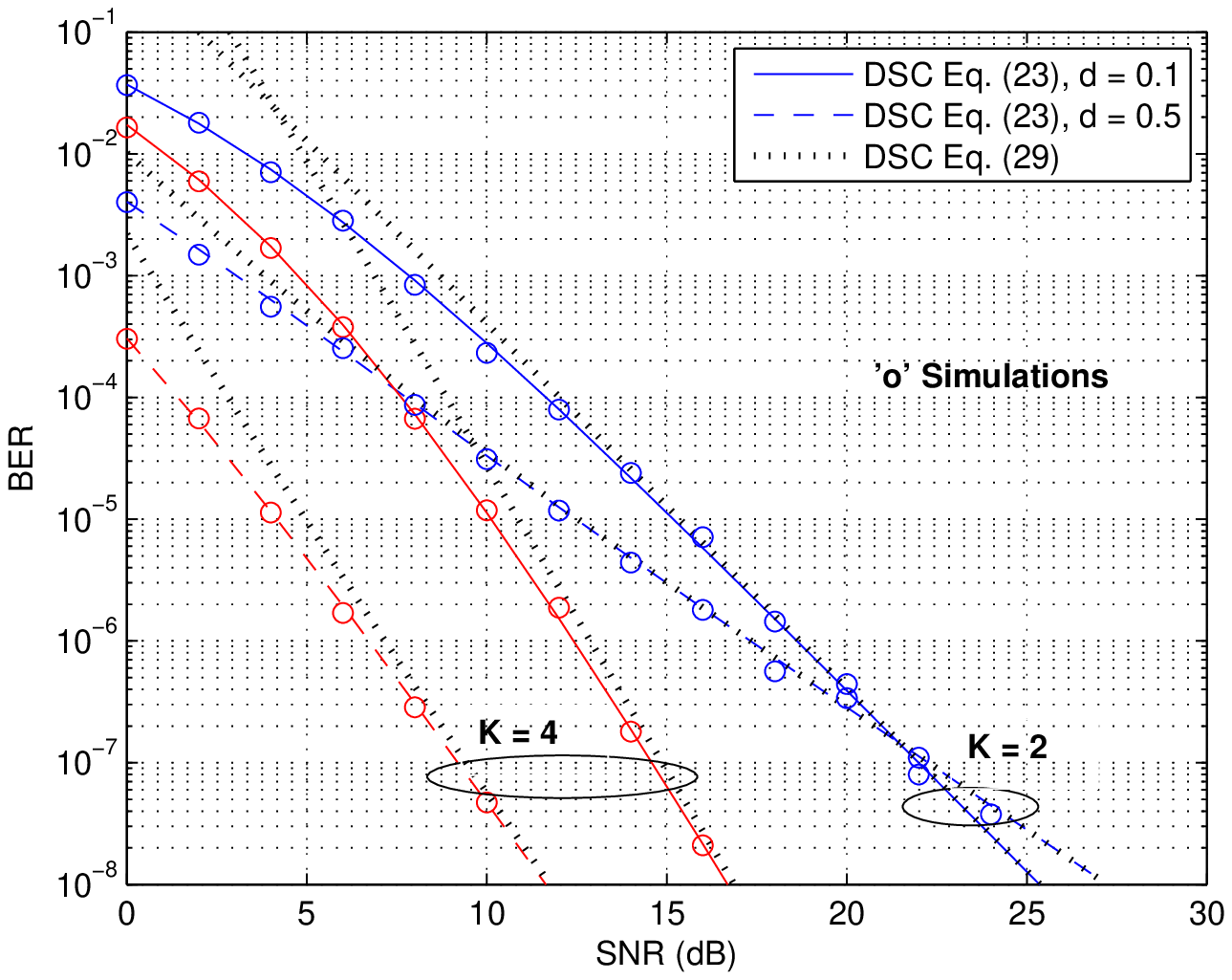}
\vspace*{-0.5 cm} \caption{End-to-end bit error rate versus SNR of the DSC scheme using a DF transmission in the LN case when $K = 2$ and $K = 4$.}\label{fig:fig2}
\end{center}
\end{figure}
\begin{figure}[htbp]
\begin{center}
\includegraphics[scale=0.75]{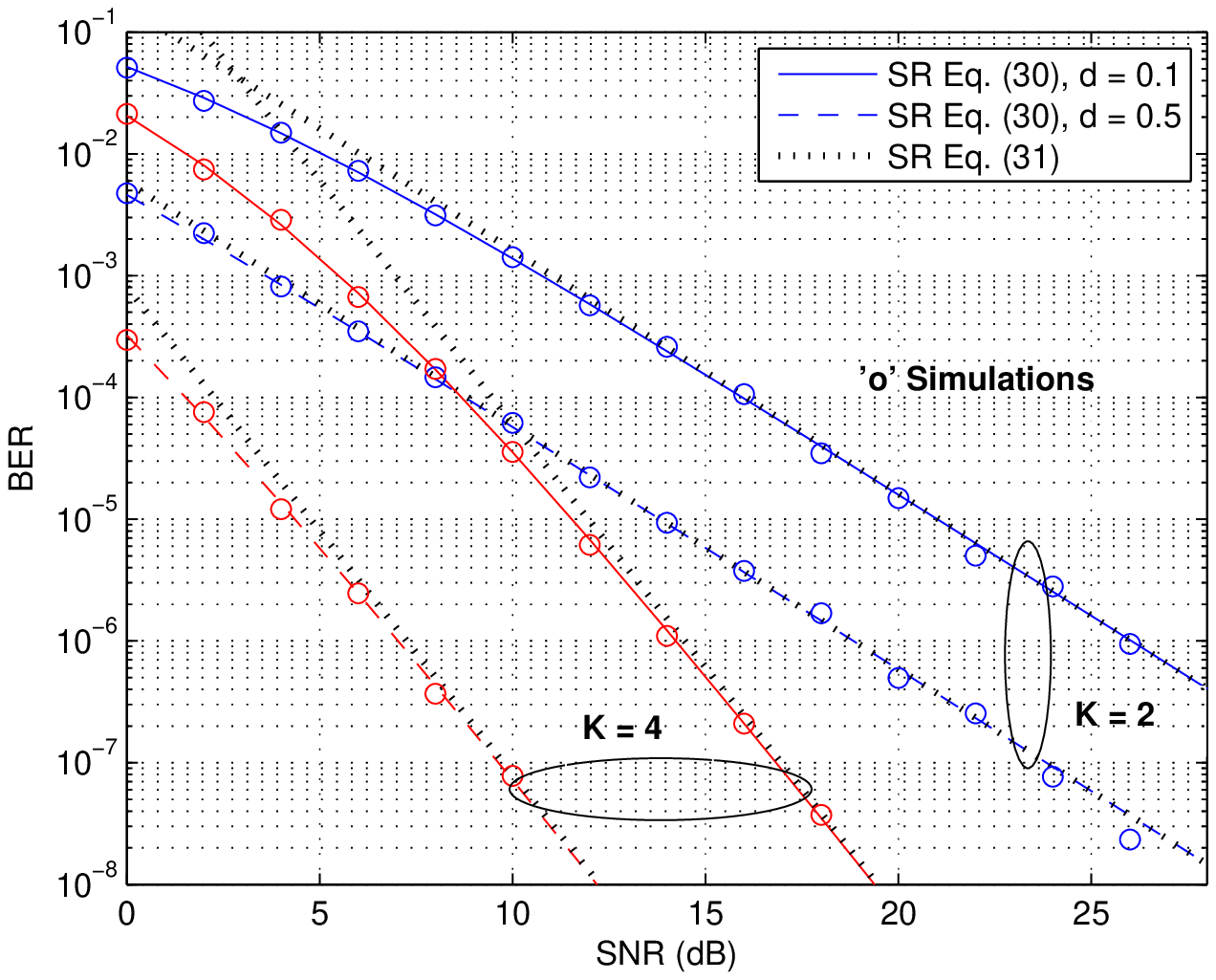}
\vspace*{-0.5 cm} \caption{End-to-end bit error-rate versus SNR of the SR scheme using a DF transmission in the LN case when $K = 2$ and $K = 4$.}\label{fig:fig3}
\end{center}
\end{figure}
\begin{figure}[htbp]
\begin{center}
\includegraphics[scale=0.75]{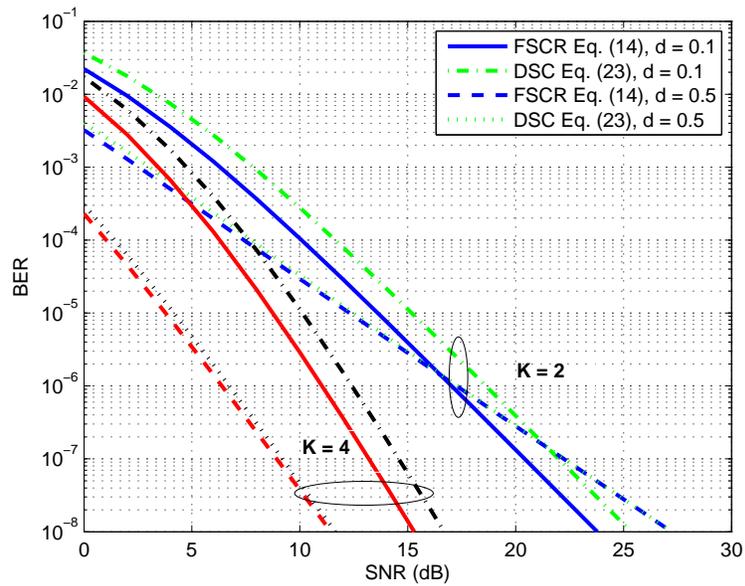}
\vspace*{-0.5 cm} \caption{Comparisons of error performance versus SNR of the FSCR and DSC schemes using a DF transmission in the LN case when $K = 2$ and $K = 4$..}\label{fig:fig4}
\end{center}
\end{figure}
\end{document}